\documentclass[10pt,twocolumn]{article}
\usepackage[letterpaper,top=0.68in,bottom=0.68in,left=0.68in,right=0.68in,columnsep=0.24in]{geometry}
\usepackage[T1]{fontenc}
\usepackage{lmodern}
\usepackage{amsmath,amssymb,mathtools,bm}
\usepackage{graphicx}
\usepackage{microtype}
\usepackage{cite}
\usepackage{etoolbox}
\usepackage{xcolor}
\usepackage{hyperref}
\hypersetup{colorlinks=true,citecolor=blue!48!black,linkcolor=blue!48!black,urlcolor=blue!48!black}

\AtBeginEnvironment{thebibliography}{\small\setlength{\itemsep}{0pt}}

\newcommand{\Tr}{\operatorname{Tr}}
\newcommand{\EF}{E_{\mathrm F}}
\newcommand{\EoFG}{E_{\mathrm F}^{\mathrm G}}
\newcommand{\Cchi}{C_{\chi}^{(1)}}
\newcommand{\CchiG}{C_{\chi,\mathrm G}^{(1)}}
\newcommand{\dd}{\mathrm d}
\newcommand{\avg}[1]{\overline{#1}}
\newcommand{\gauss}[1]{#1^{\mathrm G}}
\usepackage{booktabs}
\usepackage{array}
\newcommand{\ket}[1]{\lvert #1\rangle}

\newcommand{\ketbra}[1]{\lvert #1\rangle\!\langle #1\rvert}
\newtheorem{lemma}{Lemma}

\begin{document}

\twocolumn[
\begin{center}
{\Large\bfseries Gaussian Optimality of Energy-Constrained One-Shot Communication\\[-1mm]
through Single-Mode Bosonic Gaussian Channels\par}
\vspace{0.50em}
{\normalsize Se-Wan Ji$^{1,*}$\par}
\vspace{0.12em}
{\small $^1$The Affiliated Institute of ETRI, Daejeon 34044, Republic of Korea\\
$^*$\href{mailto:sewanji@nsr.re.kr}{sewanji@nsr.re.kr}\par}
\vspace{0.50em}
\begin{minipage}{0.92\textwidth}
\small
We prove Gaussian optimality for the energy-constrained one-shot Holevo capacity of single-mode bosonic Gaussian channels, including phase-sensitive channels with arbitrarily squeezed thermal noise.  The unresolved regime is the low-energy branch, where the Gaussian optimizer modulates only one quadrature and minimum-output-entropy arguments cannot decouple the average state from the letters.  For pure one-mode dilations we retain a stronger pointwise result: moment-matched Gaussianization improves the fixed-average Holevo function for every input.  For mixed environments, whose purification produces a $1{:}2$ entanglement-of-formation problem, we avoid any generic $1{:}2$ Gaussian extremality conjecture.  The optimal Gaussian letter selects an effective environmental Schmidt mode and an affine two-mode EPR witness.  Its null direction is exactly the modulated quadrature, while its slope equals the negative derivative of the Gaussian letter-output entropy.  This produces a supporting lower bound on the dilated entanglement of formation; Gaussian maximum entropy and concavity then give a global upper bound tangent at the Gaussian optimizer.  Consequently the known Gaussian formulas for attenuating, amplifying, phase-conjugating, and additive-noise fiducial channels are exact over unrestricted ensembles at every input energy.  Via the passive-input fiducial decomposition and energy-constrained continuity, the result extends to every single-mode Gaussian channel, including lower-rank canonical limits.  No additivity across channel uses is assumed.
\end{minipage}
\end{center}
\vspace{0.50em}
]

Bosonic Gaussian channels model optical, microwave, and hybrid linear communication links~\cite{CavesDrummond1994,HolevoWerner2001,Braunstein2005,Weedbrook2012,HolevoGiovannetti2012}.  Their energy-constrained classical information is governed by continuous-ensemble Holevo quantities~\cite{Holevo1998,Shirokov2006}.  For phase-insensitive Gaussian channels, majorization and Gaussian-optimizer theorems identify the minimum output entropy and the classical capacity~\cite{Mari2014,Giovannetti2015,DePalma2016}.  Squeezed noise breaks the symmetry.  Above a channel-dependent energy threshold, the entropy-maximizing average state and the minimum-output-entropy letter are compatible, so the usual water-filling argument is exact even for squeezed thermal noise~\cite{Schaefer2013,Schaefer2016}.  Below threshold the optimal Gaussian code leaves the noisier quadrature unmodulated, and the known Gaussian solution has remained only a lower bound on the unrestricted one-shot Holevo quantity~\cite{Schaefer2016}.

The obstruction is a coupled optimization.  At fixed average input, the Holevo information is the entropy of the average output minus the least average output entropy of the letters.  A minimum-output-entropy theorem alone does not exclude a correlated non-Gaussian choice of average and letters when the globally minimizing letter cannot be modulated within the energy constraint.  A recent affine theorem for two-mode entanglement of formation~\cite{Adesso2026} supplies the missing convex-roof control.  We first recall the pointwise pure-environment consequence, then show that the affine structure is strong enough to treat purified thermal environments without proving a generic $1{:}2$ extremality theorem.

\emph{Setting.---}
Let $\bm R=(\hat q,\hat p)^T$, $[\hat q,\hat p]=i$, define the first-moment vector $\bm d=\langle\bm R\rangle$ and $\Delta\bm R=\bm R-\bm d$, and set $V_{jk}=\langle\{\Delta R_j,\Delta R_k\}\rangle/2$.  Here $I_2$ is the $2\times2$ identity, so the vacuum covariance is $I_2/2$, and $S(\rho)=-\Tr\rho\log_2\rho$.  In general a Gaussian channel maps $\bm d\mapsto X\bm d+\bm d_0$ and $V\mapsto XVX^T+Y$; the fixed output displacement $\bm d_0$ is irrelevant to all output entropies and Holevo quantities, so we set it to zero and write
\begin{equation}
 \bm d\mapsto X\bm d,\qquad V\mapsto XVX^T+Y .
 \label{eq:channel}
\end{equation}
For an average state $\rho$, generalized pure-state ensembles define
\begin{equation}
 \chi_\Phi(\rho)=S[\Phi(\rho)]-
 \inf_{\int |\psi\rangle\!\langle\psi|\,\mu(\dd\psi)=\rho}
 \int S[\Phi(|\psi\rangle\!\langle\psi|)]\,\mu(\dd\psi),
 \label{eq:fixedchi}
\end{equation}
with entropies in bits.  Under the mean-photon constraint $\Tr\rho\hat N\le\bar N$, $\hat N=(\hat q^2+\hat p^2-1)/2$, write $\Cchi(\Phi,\bar N)=\sup\chi_\Phi(\rho)$ and let $\CchiG$ restrict the average and pure letters to Gaussian states.

For a Stinespring isometry $U:A\to BE$ and $\omega=U\rho U^\dagger$, the infinite-dimensional Matsumoto--Shimono--Winter identity gives~\cite{MSW2004,Shirokov2006}
\begin{equation}
                  \chi_\Phi(\rho)=S(\omega_B)-\EF(B{:}E)_\omega.
 \label{eq:msw}
\end{equation}
The recent affine EPR theorem implies, in particular, that for every finite-energy two-mode state $\sigma$ and its moment-matched Gaussian state $\gauss\sigma$,
\begin{equation}
                 \EF(\sigma)\ge \EF(\gauss\sigma).
 \label{eq:two-mode-ext}
\end{equation}
This yields the following stronger statement whenever the environment itself is one pure mode.

\emph{Theorem 1 (pointwise pure-environment Gaussianization).---}
Let $\Phi:A\to B$ admit a Gaussian Stinespring isometry $A\to BE$ obtained from one pure Gaussian environment mode.  For every finite-energy $\rho$, let $\gauss\rho$ have the same first and second moments.  Then
\begin{equation}
 \chi_\Phi(\rho)\le \chi_\Phi(\gauss\rho)
 =\chi_{\Phi,\mathrm G}(\gauss\rho).
 \label{eq:pointwise}
\end{equation}
Consequently Gaussian encodings attain the one-shot capacity for every confining positive-definite quadratic input cost.

Indeed, Gaussian maximum entropy~\cite{WolfExtremality2006} raises $S(\omega_B)$ under moment matching, Eq.~\eqref{eq:two-mode-ext} lowers the subtracted formation term, and an optimal Gaussian formation ensemble of $\gauss\omega$ lies in $\operatorname{Ran}U$ and pulls back to Gaussian input vectors.  The full infinite-dimensional argument is given in the Supplemental Material~\cite{SupplementalMaterial}.

We now remove the pure-environment restriction for the photon-number problem.  By a passive input rotation and an output Gaussian unitary, every nondegenerate single-mode Gaussian channel is equivalent for the energy-constrained Holevo problem to the fiducial family~\cite{Schaefer2013}
\begin{align}
 X_\tau&=\sqrt{|\tau|}\,\operatorname{diag}(1,\operatorname{sgn}\tau),\nonumber\\
 Y&=y\,\operatorname{diag}(\omega^{-1},\omega),
 \qquad y\ge\frac{|1-\tau|}{2},\quad 0<\omega\le1.
 \label{eq:fiducial}
\end{align}
We denote this fiducial channel by $\Phi_{\tau,\omega,y}$.  Let $k=|\tau|$.  For a centered one-mode Gaussian state with covariance $V$, define
\begin{align}
 h(V)&=g\!\left(\sqrt{\det V}-\tfrac12\right),\label{eq:h}\\
 g(x)&=(x+1)\log_2(x+1)-x\log_2x.
\end{align}
The known low-energy Gaussian optimizer has a pure letter
\begin{equation}
 \gamma(s)=\tfrac12\operatorname{diag}(s^{-1},s),
 \label{eq:letter}
\end{equation}
and, because the $q$ noise is larger, an average covariance
\begin{equation}
 \avg V(s,\bar N)=\operatorname{diag}\!\left(\frac{1}{2s},
 2\bar N+1-\frac{1}{2s}\right).
 \label{eq:average}
\end{equation}
Write $T=2\bar N+1$, $u=(2s)^{-1}$, and introduce
\begin{align}
 G_T(z)&=g\!\left(\sqrt{(kz+y/\omega)[k(T-z)+y\omega]}-\tfrac12\right),
 \label{eq:GT}\\
 e(u)&=g\!\left(\sqrt{(ku+y/\omega)[k/(4u)+y\omega]}-\tfrac12\right).
 \label{eq:eu}
\end{align}
Thus the low-energy Gaussian objective is $G_T(u)-e(u)$.

\emph{Lemma 2 (tangent EPR support).---}
Let $u_*$ be the interior low-energy Gaussian maximizer, $s_*=(2u_*)^{-1}>\omega$, and purify the squeezed thermal environment of Eq.~\eqref{eq:fiducial}.  For the resulting dilation $A\to BE_1E_2$, there exists a quadratic two-mode EPR witness, acting on $B$ and one Gaussian Schmidt mode of $E_1E_2$, such that every centered finite-energy input state with covariance element $V_{qq}=z$ obeys
\begin{equation}
 \EF(B{:}E_1E_2)\ge e(u_*)-\alpha_*(z-u_*),
 \qquad \alpha_*=-e'(u_*)>0.
 \label{eq:tangent-eof}
\end{equation}
The witness is the affine observable of Ref.~\cite{Adesso2026}; its gain is chosen so that the modulated $p$ direction is a null direction.  Explicitly, if
\begin{equation}
 \mu_*=\sqrt{(ku_*+y/\omega)[k/(4u_*)+y\omega]},\qquad
 q_*=\sqrt{\frac{\mu_*-1/2}{\mu_*+1/2}},
 \label{eq:muq}
\end{equation}
then the slope reduces to the channel-local expression
\begin{equation}
 \alpha_*=\frac{k y(s_*^2-\omega^2)}{\mu_*\omega}\log_2\frac1{q_*}
 =-e'(u_*).
 \label{eq:alpha}
\end{equation}
More explicitly, for the ordered EPR pair selected by the seed, the affine observable $\widehat D_{t_*}$ of Ref.~\cite{Adesso2026} satisfies
\begin{equation}
 \begin{aligned}
 \Tr(\Omega\widehat D_{t_*})-d_{t_*}(q_*)&=\beta_*(z-u_*),\\
 \lambda_{q_*,t_*}\beta_*&=\alpha_*=-e'(u_*).
 \end{aligned}
 \label{eq:tangent-mechanism}
\end{equation}
where $d_t(q)$ and $\lambda_{q,t}>0$ are the TMSV witness value and affine coefficient, respectively.  Thus the witness is insensitive to the modulated quadrature and its supporting slope is exactly the derivative required by the Gaussian optimization.  No unrestricted $1{:}2$ Gaussian entanglement-of-formation theorem is used.  The full construction, including the three fiducial dilation classes and the possible EPR-party exchange, is given in the Supplemental Material~\cite{SupplementalMaterial}, using standard Gaussian dilations and modewise decomposition~\cite{Caruso2011,Serafini2005}.

\emph{Theorem 2 (squeezed-thermal one-shot Gaussian optimizer).---}
For every fiducial channel in Eq.~\eqref{eq:fiducial} and every $\bar N<\infty$,
\begin{equation}
 \boxed{\Cchi(\Phi_{\tau,\omega,y},\bar N)
       =\CchiG(\Phi_{\tau,\omega,y},\bar N)}.
 \label{eq:thermal-theorem}
\end{equation}

\emph{Proof below threshold.---}
Common first moments may be removed because they consume photon number and only induce an output displacement.  For an arbitrary centered input covariance
$V=\begin{psmallmatrix}z&c\\c&w\end{psmallmatrix}$ with $z+w\le T$, physicality gives $zw-c^2\ge1/4$ and hence $z(T-z)\ge1/4$.  Gaussian maximum entropy gives
\begin{align}
 S[\Phi(\rho)]
 &\le g\!\left(\sqrt{(kz+y/\omega)(kw+y\omega)-k^2c^2}-\tfrac12\right)\nonumber\\
 &\le G_T(z).
 \label{eq:output-upper}
\end{align}
Combining Eq.~\eqref{eq:msw} with Lemma~2 therefore yields
\begin{equation}
 \chi_\Phi(\rho)\le
 F(z):=G_T(z)-e(u_*)+\alpha_*(z-u_*).
 \label{eq:F}
\end{equation}
The function $(kz+y/\omega)[k(T-z)+y\omega]$ is a positive concave quadratic in $z$; since both the square root and $g(x-1/2)$ are increasing and concave, $G_T$ and hence $F$ are concave.  Stationarity of the Gaussian optimization gives $G_T'(u_*)=e'(u_*)=-\alpha_*$, so $F'(u_*)=0$.  Thus $F(z)\le F(u_*)=G_T(u_*)-e(u_*)$, which is exactly the Gaussian Holevo value and is attained by the one-quadrature Gaussian displacement ensemble.  The case $\bar N=0$ is trivial.

Above threshold, the resonant letter has $s=\omega$ and minimizes output entropy globally by Gaussian-unitary equivalence to a phase-insensitive channel and the Gaussian optimizer theorem~\cite{Mari2014,Giovannetti2015,DePalma2016}.  The entropy maximum saturates the available photon budget; at that fixed output trace it is attained by an isotropic covariance, which is simultaneously feasible with the resonant minimum-entropy letter.  Hence the standard water-filling upper bound is attained.  The additive-noise line $\tau=1$ follows from the $\tau\to1$ limit using the energy-constrained Holevo-capacity continuity bound of Proposition~6 in Ref.~\cite{ShirokovECD2018}; $\tau=0$ and the identity endpoint are immediate. \hfill$\square$

The exact unrestricted capacity is therefore the known Gaussian formula for arbitrary noise amplitude $y$.  Below threshold,
\begin{equation}
 \boxed{\Cchi=\max_{s_-\le s\le s_+}
 \{h(k\avg V(s,\bar N)+Y)-h(k\gamma(s)+Y)\}},
 \label{eq:lowcapacity}
\end{equation}
where $s_\pm=(\sqrt{\bar N+1}\pm\sqrt{\bar N})^2$, and
\begin{equation}
 \bar N_{\mathrm{thr}}=
 \frac{1}{2\omega}\left[1+\frac{y}{k}(1-\omega^2)\right]-\frac12
 \qquad (k>0).
 \label{eq:threshold}
\end{equation}
Above threshold,
\begin{equation}
 \Cchi=g\!\left[k\!\left(\bar N+\tfrac12\right)+\frac{\Tr Y-1}{2}\right]
 -g\!\left(\frac{k}{2}+y-\frac12\right).
 \label{eq:highcapacity}
\end{equation}

\begin{figure*}[t]
 \centering
 \includegraphics[width=0.90\textwidth]{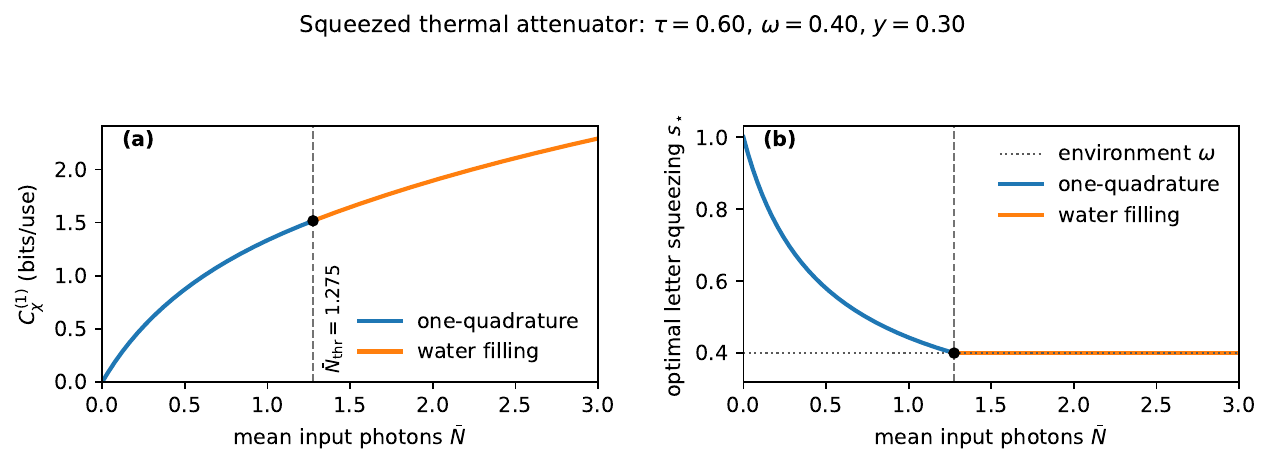}
 \caption{Exact unrestricted one-shot capacity for a genuinely thermal phase-sensitive attenuator: $\tau=0.60$, $\omega=0.40$, and $y=0.30>|1-\tau|/2=0.20$.  (a) The low-energy one-quadrature branch joins water filling at $\bar N_{\mathrm{thr}}=1.275$.  (b) The optimizing letter squeeze decreases continuously from the vacuum value toward the environment squeeze and locks to $s_\star=\omega$ at threshold.  The curves use Eqs.~\eqref{eq:lowcapacity}--\eqref{eq:highcapacity}; the Supplemental Material gives representative values and the numerical reproduction procedure.  The plotting code was prepared with the AI assistance disclosed in the Acknowledgments, and the plotted values were cross-checked against the tabulated numerical results.}
 \label{fig:thermal}
\end{figure*}

Operationally, below threshold a squeezed pure seed and a one-dimensional Gaussian displacement distribution are sufficient even when the physical environment is mixed and squeezed.  Photon subtraction, non-Gaussian seeds, arbitrarily large discrete constellations, and non-Gaussian average states cannot improve the single-use optimum under the same mean-photon constraint.  Above threshold the same seed is displaced in both quadratures.

\emph{Corollary (all single-mode Gaussian channels).---}
For every single-mode bosonic Gaussian channel $\Phi$ and every finite mean-photon constraint,
\begin{equation}
 C_\chi^{(1)}(\Phi,\bar N)=C_{\chi,\mathrm G}^{(1)}(\Phi,\bar N).
 \label{eq:all-channels}
\end{equation}
For full-rank $X$ and $Y$, the fiducial decomposition of Ref.~\cite{Schaefer2013} uses only a passive input phase rotation and an output Gaussian unitary, so Eq.~\eqref{eq:thermal-theorem} applies directly.  For a lower-rank channel choose full-rank Gaussian channels $\Phi_n\to\Phi$.  Finite-matrix convergence gives strong channel convergence, while bounded $X_n,Y_n$ give a uniform output-energy bound on every photon-number constrained input set.  Proposition~6 of Ref.~\cite{ShirokovECD2018} then gives continuity of the unrestricted constrained Holevo capacity along this limit.  Let $C_{\mathrm{disp,G}}^{(1)}$ denote the subclass of pure-Gaussian seeds with Gaussian displacement modulation.  Its covariance variational problem is continuous on a compact energy-constrained domain, and for every $\Phi_n$ the proof above gives $C_{\mathrm{disp,G}}^{(1)}(\Phi_n,\bar N)=C_\chi^{(1)}(\Phi_n,\bar N)$.  Passing to the limit and using $C_{\mathrm{disp,G}}^{(1)}\le C_{\chi,\mathrm G}^{(1)}\le C_\chi^{(1)}$ proves Eq.~\eqref{eq:all-channels}.

\emph{Scope.---}
Equation~\eqref{eq:thermal-theorem} is a single-use statement, not an additivity theorem.  For $n$ uses, MSW contains $\EF(B^n{:}E^n)$ and a single tangent witness no longer controls arbitrary cross-use correlations.  Neither the two-mode theorem of Ref.~\cite{Adesso2026} nor its bisymmetric multimode extension establishes the generic $n\times n$ formation extremality that would be needed to identify the regularized classical capacity.  The present result instead closes the full one-shot non-Gaussian gap for single-mode Gaussian communication under a mean-photon constraint, including the squeezed-thermal regime that lies outside the original two-mode Stinespring argument.

\begin{small}
\noindent\emph{Acknowledgments.---}The author used ChatGPT (GPT-5.6 Sol, OpenAI) as an auxiliary tool for checking algebraic derivations, numerical consistency tests, verification-code assistance, and manuscript editing.  All analytical arguments, numerical results, and cited claims were independently reviewed and validated by the author, who takes full responsibility for the results and conclusions.  This research was supported by the National Research Council of Science \& Technology (NST) grant by the Korea government (MSIT) (No.~CAP22053-200).

\noindent\emph{Data Availability.---}No external data sets were used.  The numerical values in Fig.~\ref{fig:thermal} are generated from the analytic expressions in the paper; the plotting and verification scripts are supplied with the Supplemental Material files.
\end{small}

\clearpage
\newgeometry{margin=0.78in}
\onecolumn
\setcounter{page}{1}
\setcounter{equation}{0}
\setcounter{figure}{0}
\setcounter{table}{0}
\setcounter{footnote}{0}
\setcounter{section}{0}
\setcounter{subsection}{0}
\setcounter{lemma}{0}
\setcounter{proposition}{0}
\setcounter{theorem}{0}
\title{Supplemental Material for\\
``Gaussian Optimality of Energy-Constrained One-Shot Communication through Single-Mode Bosonic Gaussian Channels''}
\author{Se-Wan Ji\\
\small The Affiliated Institute of ETRI, Daejeon 34044, Republic of Korea\\
\small \href{mailto:sewanji@nsr.re.kr}{sewanji@nsr.re.kr}}
\date{}
\begingroup
\fontsize{11pt}{13.2pt}\selectfont
\maketitle
\tableofcontents

\section{Conventions, energy domain, and generalized ensembles}

For an $m$-mode system let
$\bm R=(\hat q_1,\hat p_1,\ldots,\hat q_m,\hat p_m)^T$ and
$[R_j,R_k]=i\Omega_{jk}$, where
$\Omega=\bigoplus_{j=1}^m\begin{psmallmatrix}0&1\\-1&0\end{psmallmatrix}$.
The first moments and covariance matrix of a state $\rho$ are
\begin{equation}
 d_j=\Tr\rho R_j,\qquad
 V_{jk}=\frac12\Tr\rho\{R_j-d_j,R_k-d_k\}.
 \label{eq:s-covariance}
\end{equation}
Here $I_{2m}$ denotes the $2m\times2m$ identity, so the vacuum covariance is $I_{2m}/2$.  We write $S(\rho)=-\Tr\rho\log_2\rho$ and
\begin{equation}
 g(x)=(x+1)\log_2(x+1)-x\log_2x,\qquad x\ge0,
 \label{eq:s-g}
\end{equation}
All entropic logarithms are base two; $\ln$ denotes the natural logarithm.

For the pointwise pure-environment theorem we allow a confining quadratic Hamiltonian
\begin{equation}
 H=\frac12\bm R^T G\bm R+\bm h^T\bm R+c_H,\qquad G>0,
 \label{eq:s-hamiltonian}
\end{equation}
where $G$ is real symmetric positive definite, $\bm h$ is real, and $c_H\in\mathbb R$.  This Hamiltonian is bounded below and satisfies the Gibbs condition $\Tr e^{-\beta H}<\infty$ for every $\beta>0$.  A finite-energy bound then controls first and second moments, and the relevant entropies are finite and continuous on the constrained set.  If $\gauss\rho$ is the Gaussian state with the same first and second moments as $\rho$, then
\begin{equation}
                         \Tr\rho H=\Tr\gauss\rho H.
 \label{eq:s-cost}
\end{equation}
For the thermal-noise extension the cost is the photon number
$\hat N=(\hat q^2+\hat p^2-1)/2$.

A generalized ensemble is a Borel probability measure $\mu$ on the pure-state space with barycenter
\begin{equation}
                         \rho=\int\ketbra\psi\,\mu(\dd\psi).
\end{equation}
The continuous convex roof is the appropriate definition of entanglement of formation and of the convex closure of output entropy in infinite dimension~\cite{Shirokov2006}.

\section{Fixed-average Holevo information and the MSW identity}

Let $\Phi:A\to B$ have Stinespring isometry $U:A\to BE$.  Define
\begin{equation}
 \widehat S_\Phi(\rho)=
 \inf_{\int\ketbra\psi\,\mu(\dd\psi)=\rho}
 \int S[\Phi(\ketbra\psi)]\,\mu(\dd\psi).
 \label{eq:s-convex-closure}
\end{equation}
Refining mixed letters into pure states cannot increase the average output entropy.  For $\omega_{BE}=U\rho U^\dagger$, the generalized fixed-average Holevo function and its infinite-dimensional MSW representation are understood in the continuous-ensemble sense~\cite{MSW2004,Shirokov2006}.  The fixed-average Holevo function satisfies
\begin{equation}
 \chi_\Phi(\rho)=S[\Phi(\rho)]-\widehat S_\Phi(\rho).
 \label{eq:s-fixed-holevo}
\end{equation}

\begin{lemma}[Infinite-dimensional fixed-average MSW identity]
For $\omega_{BE}=U\rho U^\dagger$,
\begin{equation}
                       \widehat S_\Phi(\rho)=\EF(B{:}E)_\omega.
 \label{eq:s-msw}
\end{equation}
\end{lemma}

\emph{Proof.}
For every pure input $\ket\psi_A$, $U\ket\psi$ is pure and
$S[\Phi(\ketbra\psi)]=\mathcal E(U\ket\psi)$, where $\mathcal E$ is pure-state entanglement across $B{:}E$.  Thus every input ensemble maps to a decomposition of $\omega$ supported on $\mathcal K=\operatorname{Ran}U$.
Conversely, if $P_{\mathcal K}=UU^\dagger$ and
$\omega=\int\ketbra\varphi\,\nu(\dd\varphi)$, then
\begin{equation}
 0=\Tr[(I-P_{\mathcal K})\omega]
  =\int\|(I-P_{\mathcal K})\ket\varphi\|^2\nu(\dd\varphi).
\end{equation}
The nonnegative integrand vanishes almost everywhere, so $U^\dagger\ket\varphi$ defines an input ensemble with barycenter $\rho$.  The two infima coincide. \hfill$\square$

Hence
\begin{equation}
                    \chi_\Phi(\rho)=S(\omega_B)-\EF(B{:}E)_\omega.
 \label{eq:s-fixed-msw}
\end{equation}

\section{Affine two-mode EPR witness}

For each mode $J$, define $\hat a_J=(\hat q_J+i\hat p_J)/\sqrt2$, and for a pure bipartite state $\ket\psi$ write $\mathcal E(\psi)$ for the entropy of either reduced state.  For two modes $A,B$, define
\begin{equation}
 \widehat D_t=(t\hat a_A-\hat a_B^\dagger)^\dagger
                    (t\hat a_A-\hat a_B^\dagger)
 =\frac12\left[(t\hat q_A-\hat q_B)^2+(t\hat p_A+\hat p_B)^2+1-t^2\right]
 \label{eq:s-epr}
\end{equation}
for $0<t\le1$.  Let
\begin{equation}
 \ket{\Psi_q}=\sqrt{1-q^2}\sum_{n=0}^\infty q^n\ket{n,n},
 \qquad 0\le q<1,
\end{equation}
which is the two-mode squeezed vacuum (TMSV) with parameter $q$.  Its entanglement is
\begin{equation}
 \mathcal E(q)=g\!\left(\frac{q^2}{1-q^2}\right),\qquad
 d_t(q)=\frac{(1-tq)^2}{1-q^2}.
\end{equation}
Ref.~\cite{Adesso2026} proves that for every finite-energy pure two-mode state and all $0<q<t\le1$,
\begin{equation}
 \mathcal E(\psi)+\lambda_{q,t}
       [\langle\widehat D_t\rangle_\psi-d_t(q)]\ge\mathcal E(q),
 \label{eq:s-affine}
\end{equation}
where in bits
\begin{equation}
 \lambda_{q,t}=\frac{-2q\ln q}{(t-q)(1-tq)\ln2}>0.
 \label{eq:s-lambda}
\end{equation}
At $q=t$ the endpoint is understood in the sense stated in Ref.~\cite{Adesso2026}.
Since Eq.~\eqref{eq:s-affine} is affine in the state-dependent expectation, averaging gives for every two-mode state $\rho$
\begin{equation}
 \EF(\rho)\ge\mathcal E(q)-\lambda_{q,t}
 [\Tr(\rho\widehat D_t)-d_t(q)]
 \label{eq:s-mixed-witness}
\end{equation}
for each admissible pair $(q,t)$.  Optimizing over $(q,t)$ recovers the mixed-state witness of Ref.~\cite{Adesso2026}.

The same result also implies fixed-covariance Gaussian extremality.  If $\sigma$ is an arbitrary finite-energy two-mode state and $\gauss\sigma$ has identical first moments and covariance, then
\begin{equation}
                         \EF(\sigma)\ge\EF(\gauss\sigma).
 \label{eq:s-eofext}
\end{equation}
Indeed, the covariance geometry of Ref.~\cite{Adesso2026} gives local symplectic coordinates in which
$V'=V_q+N_{\rm cl}$, where $V_q$ is the covariance of $\ket{\Psi_q}$ and $N_{\rm cl}\ge0$ is a classical-displacement covariance orthogonal to the EPR quadratures in Eq.~\eqref{eq:s-epr}.  Thus $\Tr\gauss\sigma\widehat D_t=d_t(q)$, while $N_{\rm cl}$ supplies a Gaussian displacement decomposition of $\gauss\sigma$ into copies of $\ket{\Psi_q}$.  The affine lower bound and this Gaussian upper decomposition meet at $\mathcal E(q)$.

\section{Pure one-mode dilation: pointwise Gaussianization and pullback}

Let $U:A\to BE$ be a Gaussian isometry obtained from one pure Gaussian environment mode.  For arbitrary $\rho$ and its moment-matched Gaussian state $\gauss\rho$, set
\begin{equation}
 \omega=U\rho U^\dagger,\qquad \gauss\omega=U\gauss\rho U^\dagger.
\end{equation}
Gaussianity of the isometry makes $\gauss\omega$ the Gaussian state with the moments of $\omega$.  Gaussian maximum entropy~\cite{WolfExtremality2006} and Eq.~\eqref{eq:s-eofext} give
\begin{equation}
 S(\omega_B)\le S(\gauss\omega_B),\qquad
 \EF(\omega_{BE})\ge\EF(\gauss\omega_{BE}),
\end{equation}
so by Eq.~\eqref{eq:s-fixed-msw}
\begin{equation}
                         \chi_\Phi(\rho)\le\chi_\Phi(\gauss\rho).
 \label{eq:s-gaussianization}
\end{equation}
Equation~\eqref{eq:s-cost} preserves feasibility for every confining positive-definite quadratic cost.

To prove Gaussian achievability at $\gauss\rho$, use $\EF(\gauss\omega)=\EoFG(\gauss\omega)$ from Ref.~\cite{Adesso2026}.  An optimal Gaussian formation ensemble lies in $\operatorname{Ran}U$ almost everywhere by the support argument in the MSW lemma.  If a pure Gaussian vector $\ket\phi_{BE}\in\operatorname{Ran}U$, then $U^\dagger\ket\phi$ is Gaussian: write $U\ket\psi=W_G(\ket\psi\otimes\ket e)$ with pure Gaussian $\ket e$ and Gaussian unitary $W_G$; then $W_G^\dagger\ket\phi=\ket\psi\otimes\ket e$ is Gaussian, hence so is its $A$ marginal.  Pulling back the optimal formation ensemble proves
\begin{equation}
 \chi_\Phi(\rho)\le\chi_\Phi(\gauss\rho)
 =\chi_{\Phi,\mathrm G}(\gauss\rho).
\end{equation}
This is the pointwise statement quoted as Theorem 1 in the Letter.

\section{Fiducial channels and squeezed-thermal dilations}

The standard fiducial representation and its one-mode Gaussian dilations are described in Refs.~\cite{Schaefer2013,Caruso2011}.  The fiducial channel is
\begin{equation}
 X_\tau=\sqrt{k}\,\operatorname{diag}(1,\operatorname{sgn}\tau),\qquad
 Y=y\operatorname{diag}(\omega^{-1},\omega),
 \label{eq:s-fiducial}
\end{equation}
We denote it by $\Phi_{\tau,\omega,y}\equiv\Phi^{\rm F}_{\tau,\omega,y}$, using the superscript only when the fiducial nature needs emphasis.  Here $k=|\tau|$, $0<\omega\le1$, and complete positivity requires
$y\ge r/2$ with
\begin{equation}
                             r=|1-\tau|.
\end{equation}
For $\tau\ne1$, write
\begin{equation}
 \nu=\frac yr\ge\frac12,\qquad
 \zeta=\sqrt{\nu^2-\frac14}.
\end{equation}
The physical environment mode $F$ may be taken in the squeezed thermal state
$V_F=\nu\operatorname{diag}(\omega^{-1},\omega)$ and purified by a mode $E_2$ with covariance
\begin{equation}
 V_{F E_2}=\begin{pmatrix}
 \nu/\omega&0&\zeta/\sqrt\omega&0\\
 0&\nu\omega&0&-\zeta\sqrt\omega\\
 \zeta/\sqrt\omega&0&\nu&0\\
 0&-\zeta\sqrt\omega&0&\nu
 \end{pmatrix}.
 \label{eq:s-purification}
\end{equation}
The canonical two-mode unitary coupling $A$ and $F$ can be chosen as follows, with $E_1$ the other output of that unitary.

For $0<\tau<1$ ($k=\tau$, $r=1-\tau$),
\begin{align}
 \hat q_B&=\sqrt k\,\hat q_A+\sqrt r\,\hat q_F,&
 \hat p_B&=\sqrt k\,\hat p_A+\sqrt r\,\hat p_F,\nonumber\\
 \hat q_{E_1}&=-\sqrt r\,\hat q_A+\sqrt k\,\hat q_F,&
 \hat p_{E_1}&=-\sqrt r\,\hat p_A+\sqrt k\,\hat p_F.
 \label{eq:s-atten}
\end{align}
For $\tau>1$ ($k=\tau$, $r=\tau-1$),
\begin{align}
 \hat q_B&=\sqrt k\,\hat q_A+\sqrt r\,\hat q_F,&
 \hat p_B&=\sqrt k\,\hat p_A-\sqrt r\,\hat p_F,\nonumber\\
 \hat q_{E_1}&=\sqrt r\,\hat q_A+\sqrt k\,\hat q_F,&
 \hat p_{E_1}&=-\sqrt r\,\hat p_A+\sqrt k\,\hat p_F.
 \label{eq:s-amp}
\end{align}
For $\tau=-k<0$ ($r=k+1$), the channel output is the idler output of the same two-mode squeezer:
\begin{align}
 \hat q_B&=\sqrt k\,\hat q_A+\sqrt r\,\hat q_F,&
 \hat p_B&=-\sqrt k\,\hat p_A+\sqrt r\,\hat p_F,\nonumber\\
 \hat q_{E_1}&=\sqrt r\,\hat q_A+\sqrt k\,\hat q_F,&
 \hat p_{E_1}&=\sqrt r\,\hat p_A-\sqrt k\,\hat p_F.
 \label{eq:s-conj}
\end{align}
Equations~\eqref{eq:s-atten}--\eqref{eq:s-conj}, together with the untouched purifier $E_2$, define a pure three-mode Stinespring dilation for every $\tau\ne1$ and $y\ge r/2$.  At the quantum-limited value $y=r/2$, $\zeta=0$ and $E_2$ becomes an irrelevant vacuum spectator.

\section{Effective environmental Schmidt mode}

Take the pure Gaussian letter
\begin{equation}
 \gamma(s)=\begin{pmatrix}u&0\\0&v\end{pmatrix},\qquad
 u=\frac1{2s},\quad v=\frac s2,\quad uv=\frac14.
 \label{eq:s-seed}
\end{equation}
The global output $BE_1E_2$ of Eqs.~\eqref{eq:s-purification}--\eqref{eq:s-conj} is pure.  The system covariance is
\begin{equation}
 V_B(s)=\begin{pmatrix}a&0\\0&b\end{pmatrix},\qquad
 a=ku+\frac y\omega,\quad b=kv+y\omega.
 \label{eq:s-ab}
\end{equation}
Define
\begin{equation}
 \mu=\sqrt{ab},\qquad c=\sqrt{\mu^2-\frac14},\qquad
 \xi=\left(\frac ba\right)^{1/4}.
 \label{eq:s-mucxi}
\end{equation}
After the local squeeze
$\hat q_{B'}=\xi\hat q_B$, $\hat p_{B'}=\hat p_B/\xi$, the $B'$ covariance is $\mu I_2$.
Gaussian Schmidt (modewise) decomposition~\cite{Serafini2005} therefore identifies one environmental mode $E$ such that $B'E$ is a TMSV and the remaining environmental mode $G$ is a vacuum spectator.  Its TMSV parameter is
\begin{equation}
 q=\sqrt{\frac{\mu-1/2}{\mu+1/2}}
   =\frac{\mu-1/2}{c},
 \qquad \mathcal E(q)=g(\mu-1/2).
 \label{eq:s-q}
\end{equation}

For later use, one can construct the effective mode directly.  For observables $A,B$ define $\operatorname{Cov}(A,B)=\langle\{\Delta A,\Delta B\}\rangle/2$.  Let
\begin{equation}
 \bm x=(\operatorname{Cov}(\hat q_{B'},\hat q_{E_1}),\operatorname{Cov}(\hat q_{B'},\hat q_{E_2})),\qquad
 \bm z=(\operatorname{Cov}(\hat p_{B'},\hat p_{E_1}),\operatorname{Cov}(\hat p_{B'},\hat p_{E_2})).
\end{equation}
In the ordering $(q_{B'},q_{E_1},q_{E_2},p_{B'},p_{E_1},p_{E_2})$, the covariance of the pure seed output has no $q$--$p$ block and may be written $V=Q\oplus P$, with
\begin{equation}
 Q=\begin{pmatrix}\mu&\bm x^T\\ \bm x&Q_{\rm env}\end{pmatrix},\qquad
 P=\begin{pmatrix}\mu&\bm z^T\\ \bm z&P_{\rm env}\end{pmatrix}.
 \label{eq:s-QP-blocks}
\end{equation}
Purity, $V\Omega V=\Omega/4$, is equivalent here to $QP=I_3/4$ and gives
\begin{equation}
 \bm x\!\cdot\!\bm z=-c^2,\qquad
 Q_{\rm env}\bm z=-\mu\bm x,\qquad
 P_{\rm env}\bm x=-\mu\bm z.
 \label{eq:s-QP-identities}
\end{equation}
For $c>0$ define
\begin{equation}
 \hat q_E=-\frac{\bm z\cdot(\hat q_{E_1},\hat q_{E_2})}{c},\qquad
 \hat p_E= \frac{\bm x\cdot(\hat p_{E_1},\hat p_{E_2})}{c}.
 \label{eq:s-effective-mode}
\end{equation}
Equation~\eqref{eq:s-QP-identities} gives $[\hat q_E,\hat p_E]=i$ and
\begin{equation}
 \operatorname{Var}(q_E)=\frac{\bm z^TQ_{\rm env}\bm z}{c^2}=\mu,\qquad
 \operatorname{Var}(p_E)=\frac{\bm x^TP_{\rm env}\bm x}{c^2}=\mu.
 \label{eq:s-effective-variances}
\end{equation}
Moreover $\operatorname{Cov}(\hat q_{B'},\hat q_E)=c$ and
$\operatorname{Cov}(\hat p_{B'},\hat p_E)=-c$.  Hence
\begin{equation}
 V_{B'E}=\begin{pmatrix}
 \mu&0&c&0\\0&\mu&0&-c\\c&0&\mu&0\\0&-c&0&\mu
 \end{pmatrix},
 \label{eq:s-explicit-tmsv}
\end{equation}
which is exactly the TMSV covariance.  Completing this canonical pair gives the spectator $G$.  Every strictly below-threshold seed used later has $c>0$; the isolated product point with $c=0$ is treated by continuity.

Introduce the two positive quantities
\begin{equation}
                 \ell=u\xi^2,\qquad m=\frac v{\xi^2},\qquad \ell m=\frac14.
 \label{eq:s-ellm}
\end{equation}
For the low-energy Gaussian branch one has $s\ge\omega$, with strict inequality below threshold.  This is equivalent to
\begin{equation}
 \ell\le\frac12\le m,
 \label{eq:s-ell-order}
\end{equation}
because $\ell\le1/2$ is equivalent to $b\le as^2$, and
\begin{equation}
                         as^2-b=\frac y\omega(s^2-\omega^2).
\end{equation}

\section{Displacement gains and the admissible Adesso witness}

The effective mode in Eq.~\eqref{eq:s-effective-mode} also tracks how input displacements are distributed between $B'$ and $E$.  Define $t_q$ so that the $q_A$ displacement cancels in $t_q q_{B'}-q_E$, and $t_p$ so that the $p_A$ displacement cancels in $t_p p_{B'}+p_E$.  Direct substitution of Eqs.~\eqref{eq:s-atten}--\eqref{eq:s-effective-mode} gives a particularly simple result.

For the covariant cases $\tau>0$ (attenuator and amplifier),
\begin{equation}
                  t_p=\frac{\mu-\ell}{c},\qquad
                  t_q=\frac{\mu-m}{c}.
 \label{eq:s-t-cov}
\end{equation}
For the contravariant phase-conjugating case $\tau<0$,
\begin{equation}
                  t_p=\frac{\mu+\ell}{c},\qquad
                  t_q=\frac{\mu+m}{c}.
 \label{eq:s-t-conj}
\end{equation}
For example, the attenuator has
$t_p=\xi^2(y/\omega-r u)/c$ and the amplifier has
$t_p=\xi^2(y/\omega+r u)/c$; the identities $k+r=1$ and $k-r=1$, respectively, reduce both expressions to Eq.~\eqref{eq:s-t-cov}.  The phase conjugator uses $r-k=1$ and gives Eq.~\eqref{eq:s-t-conj}.

The raw gain $t_p$ always yields an admissible Adesso gain after at most exchanging the two EPR parties.  For $\tau>0$,
\begin{equation}
 t_p-q=\frac{1/2-\ell}{c}\ge0,\qquad
 t_p<\frac1q=\frac{\mu+1/2}{c}.
 \label{eq:s-gain-cov}
\end{equation}
Thus if $t_p\le1$ we use the operator $\widehat D_{t_p}$ of Eq.~\eqref{eq:s-epr} with the ordered parties $(A,B)=(B',E)$; if $t_p>1$ we exchange the parties and use $\widehat D_{1/t_p}$ with $(A,B)=(E,B')$, for which $q<1/t_p<1$ below threshold.  For $\tau<0$, $t_p=(\mu+\ell)/c>1$, but Eq.~\eqref{eq:s-ell-order} gives
\begin{equation}
                       t_p\le\frac1q,
\end{equation}
so the exchanged gain $1/t_p$ again lies in $(q,1)$.  For $c>0$, as $s\downarrow\omega$ the admissible gain reaches the endpoint $q=t$ (after the appropriate exchange in the contravariant case).  In the quantum-limited attenuator the resonant point instead has $\mu=1/2$ and hence $c=q=0$, so the raw gain ratios are of $0/0$ form there; this endpoint is understood only by the $s\downarrow\omega$ limit and is never evaluated in the strictly below-threshold proof.

\section{A \texorpdfstring{$1{:}2$}{1:2} tangent formation bound from a two-mode witness}

The Gaussian Schmidt reduction is used only to select a two-mode witness; no Gaussian optimality theorem for generic $1{:}2$ states is assumed.
Let $\Omega_{B'EG}$ be an arbitrary finite-energy state produced by the same Stinespring dilation, after the fixed local Gaussian transformations chosen from the Gaussian seed at $s$.  We first record explicitly the monotonicity step needed below:
\begin{equation}
 \EF(B'{:}EG)_\Omega\ge\EF(B'{:}E)_{\Tr_G\Omega}.
 \label{eq:s-monotonicity}
\end{equation}
To see this directly in the continuous-convex-roof setting, take any generalized pure-state decomposition $\Omega=\int\ketbra\psi\,\mu(\dd\psi)$ and put $\rho_{B'E}^{\psi}=\Tr_G\ketbra\psi$.  For every $\psi$, concavity of entropy and the convex-roof definition give
\begin{equation}
 \EF(B'{:}E)_{\rho^{\psi}}\le S(\rho^{\psi}_{B'})=\mathcal E(B'{:}EG)_\psi .
 \label{eq:s-trace-pure}
\end{equation}
Convexity of $\EF$ then yields
\begin{equation}
 \EF(B'{:}E)_{\Tr_G\Omega}
 \le\int \EF(B'{:}E)_{\rho^{\psi}}\,\mu(\dd\psi)
 \le\int \mathcal E(B'{:}EG)_\psi\,\mu(\dd\psi).
\end{equation}
Taking the infimum over all decompositions of $\Omega$ proves Eq.~\eqref{eq:s-monotonicity}.
Applying Eq.~\eqref{eq:s-mixed-witness} to the reduced two-mode state gives
\begin{equation}
 \EF(B'{:}EG)_\Omega\ge
 \mathcal E(q)-\lambda_{q,t}
 [\Tr(\Omega\widehat D_t)-d_t(q)].
 \label{eq:s-lifted-witness}
\end{equation}
Here $t=t_p$ if the raw gain is at most one; otherwise $t=1/t_p$ and the EPR parties are exchanged.  In Eq.~\eqref{eq:s-lifted-witness}, $\widehat D_t$ always denotes the correspondingly ordered two-mode operator specified in Sec.~7.

Now let the arbitrary centered input have covariance
\begin{equation}
 V=\begin{pmatrix}z&c_0\\c_0&w_0\end{pmatrix}.
 \label{eq:s-arbitrary-V}
\end{equation}
The witness was chosen so that its $p$ EPR quadrature has no response to an input $p$ displacement.  Its $q$ EPR quadrature has response proportional to $t_p-t_q$.  Since $\widehat D_t$ is quadratic and the environment state is fixed, comparison with the seed covariance in Eq.~\eqref{eq:s-seed} gives exactly
\begin{equation}
 \Tr(\Omega\widehat D_t)-d_t(q)=\beta_{\rm W}(z-u),
 \label{eq:s-beta-def}
\end{equation}
with no dependence on $w_0$ or $c_0$.  If the raw gain is used,
\begin{equation}
             \beta_{\rm W}=\frac{k\xi^2}{2}(t_p-t_q)^2,
 \label{eq:s-beta-raw}
\end{equation}
and after exchanging the EPR parties,
\begin{equation}
             \beta_{\rm W}=\frac{k\xi^2}{2t_p^2}(t_p-t_q)^2.
 \label{eq:s-beta-swap}
\end{equation}
The Gaussian seed itself has $B'E$ exactly in the TMSV state, so the reference expectation in Eq.~\eqref{eq:s-beta-def} is $d_t(q)$.

Let
\begin{equation}
                         \alpha=\lambda_{q,t}\beta_{\rm W}.
 \label{eq:s-alpha-def}
\end{equation}
A useful fact is that the two gain branches and the possible EPR-party exchange all reduce to the same channel-local slope.  First,
\begin{equation}
 m-\ell=\frac{va-ub}{\mu}
       =\frac{y(s^2-\omega^2)}{2s\omega\mu}.
 \label{eq:s-mell}
\end{equation}
For the covariant raw-gain case, Eqs.~\eqref{eq:s-t-cov}, \eqref{eq:s-q}, and \eqref{eq:s-ellm} give
\begin{equation}
 t_p-q=\frac{1/2-\ell}{c},\qquad
 1-t_pq=\frac{1/2+\ell}{\mu+1/2},
\end{equation}
and hence
\begin{equation}
 \lambda_{q,t_p}=\frac{8c^2}{1-4\ell^2}\log_2\frac1q.
\end{equation}
Using $t_p-t_q=(m-\ell)/c$, Eq.~\eqref{eq:s-beta-raw}, $m=1/(4\ell)$, and $\xi^2=\ell/u=2s\ell$ yields
\begin{equation}
 \alpha=2ks(m-\ell)\log_2\frac1q.
 \label{eq:s-alpha-intermediate}
\end{equation}
If $t_p>1$, the admissible witness uses $t'=1/t_p$ and Eq.~\eqref{eq:s-beta-swap}.  Its affine coefficient satisfies the exact identity
\begin{equation}
 \lambda_{q,1/t_p}=t_p^2\lambda_{q,t_p},
 \label{eq:s-lambda-swap}
\end{equation}
because $(1/t_p-q)(1-q/t_p)=(t_p-q)(1-t_pq)/t_p^2$.  Hence the factor $t_p^2$ cancels the $t_p^{-2}$ in Eq.~\eqref{eq:s-beta-swap}, so the raw- and exchanged-party products $\lambda\beta_{\rm W}$ coincide.

For the contravariant class, $t_p=(\mu+\ell)/c$ and $t_q=(\mu+m)/c$, so
\begin{equation}
 t_p-q=\frac{1/2+\ell}{c},\qquad
 1-t_pq=\frac{1/2-\ell}{\mu+1/2},\qquad
 (t_p-t_q)^2=\frac{(m-\ell)^2}{c^2}.
 \label{eq:s-conj-algebra}
\end{equation}
The first two relations give the same algebraic coefficient
\begin{equation}
 \lambda_{q,t_p}=\frac{8c^2}{1-4\ell^2}\log_2\frac1q
\end{equation}
as in the covariant branch.  Although $t_p>1$ is not itself an admissible Adesso gain, Eq.~\eqref{eq:s-lambda-swap} shows that the actual exchanged witness $t=1/t_p$ multiplied by Eq.~\eqref{eq:s-beta-swap} has exactly this product.  Together with the last relation in Eq.~\eqref{eq:s-conj-algebra}, this again gives Eq.~\eqref{eq:s-alpha-intermediate}.  Inserting Eq.~\eqref{eq:s-mell},
\begin{equation}
 \boxed{\alpha=\frac{k y(s^2-\omega^2)}{\mu\omega}\log_2\frac1q.}
 \label{eq:s-alpha}
\end{equation}
Therefore Eq.~\eqref{eq:s-lifted-witness} becomes
\begin{equation}
 \boxed{\EF(B{:}E_1E_2)_\Omega
 \ge g(\mu-1/2)-\alpha(z-u).}
 \label{eq:s-tangent-bound}
\end{equation}
Local Gaussian transformations have been undone in writing the left-hand side.

The word ``tangent'' is literal.  Define the pure Gaussian letter-output entropy as a function of its $q$ variance,
\begin{equation}
 e(u)=g\!\left(\sqrt{(ku+y/\omega)[k/(4u)+y\omega]}-\frac12\right).
 \label{eq:s-eu}
\end{equation}
With $s=(2u)^{-1}$ and $\mu$ from Eq.~\eqref{eq:s-mucxi},
\begin{equation}
 \frac{\dd\mu}{\dd u}=
 \frac{k}{2\mu}(b-a s^2)
 =-\frac{k y(s^2-\omega^2)}{2\mu\omega},
\end{equation}
and
\begin{equation}
 \frac{\dd}{\dd\mu}g(\mu-1/2)
 =\log_2\frac{\mu+1/2}{\mu-1/2}
 =2\log_2\frac1q.
\end{equation}
Consequently
\begin{equation}
                              \boxed{\alpha=-e'(u).}
 \label{eq:s-tangent-identity}
\end{equation}
For $s>\omega$, $e'(u)<0$ and $\alpha>0$.  At resonance $s=\omega$ the slope vanishes.

\section{Global supporting-hyperplane proof below threshold}

Let $T=2\bar N+1$.  For a Gaussian one-quadrature encoding with seed $u$,
\begin{equation}
 \gamma_u=\operatorname{diag}\!\left(u,\frac1{4u}\right),\qquad
 \avg V_u=\operatorname{diag}(u,T-u).
\end{equation}
Feasibility is
\begin{equation}
                      T-u-\frac1{4u}\ge0.
 \label{eq:s-u-feasible}
\end{equation}
Equivalently, in the squeezing variable $s=(2u)^{-1}$, $s_-\le s\le s_+$ with
$s_\pm=(\sqrt{\bar N+1}\pm\sqrt{\bar N})^2$.
The Gaussian low-energy objective is
\begin{equation}
 C_{\chi,\mathrm G}^{(1)}=\max_u\{G_T(u)-e(u)\},
 \label{eq:s-gauss-objective}
\end{equation}
where
\begin{equation}
 G_T(z)=g\!\left(\sqrt{(kz+y/\omega)[k(T-z)+y\omega]}-\frac12\right).
 \label{eq:s-GT}
\end{equation}
For $\bar N>0$ and a nontrivial channel $k>0$, the endpoints of Eq.~\eqref{eq:s-u-feasible} have zero modulation and hence zero Holevo information, while an interior displacement encoding has positive information.  Thus the low-energy maximizer $u_*$ is interior.  The known Gaussian solution~\cite{Schaefer2013,Schaefer2016} has $s_*=(2u_*)^{-1}>\omega$ strictly below threshold and reaches $s_*=\omega$ at threshold.

Now take an arbitrary centered input state with covariance Eq.~\eqref{eq:s-arbitrary-V} and $z+w_0\le T$.  Physicality gives $zw_0-c_0^2\ge1/4$, and therefore
\begin{equation}
 z(T-z)\ge zw_0\ge\frac14+c_0^2\ge\frac14.
 \label{eq:s-z-domain}
\end{equation}
Thus $z$ lies in the compact interval
\begin{equation}
 I_T=\left[\frac{T-\sqrt{T^2-1}}2,\frac{T+\sqrt{T^2-1}}2\right],
 \label{eq:s-IT}
\end{equation}
and $\operatorname{diag}(z,T-z)$ is itself a physical covariance.  The arbitrary input output covariance has determinant
\begin{equation}
 \det V_B=(kz+y/\omega)(kw_0+y\omega)-k^2c_0^2.
\end{equation}
Gaussian maximum entropy gives
\begin{align}
 S[\Phi(\rho)]
 &\le g\!\left(\sqrt{(kz+y/\omega)(kw_0+y\omega)-k^2c_0^2}-\frac12\right)\nonumber\\
 &\le G_T(z).
 \label{eq:s-output-upper}
\end{align}
The second inequality follows because the expression increases with $w_0$ and decreases with $c_0^2$, so it is maximized at $w_0=T-z$ and $c_0=0$.  First moments can be removed by a common displacement, which leaves all output entropies and Holevo information invariant while reducing the photon cost.

Construct the tangent witness of the previous section at the Gaussian optimizer $u_*$.  By Eqs.~\eqref{eq:s-fixed-msw}, \eqref{eq:s-tangent-bound}, and \eqref{eq:s-tangent-identity}, every input obeys
\begin{equation}
 \chi_\Phi(\rho)\le
 F(z):=G_T(z)-e(u_*)+\alpha_*(z-u_*),
 \qquad \alpha_*=-e'(u_*).
 \label{eq:s-F}
\end{equation}
The function
\begin{equation}
 D_T(z)=(kz+y/\omega)[k(T-z)+y\omega]
\end{equation}
is a concave quadratic.  For every $z\in I_T$, it is the determinant of the output covariance generated by the physical Gaussian input $\operatorname{diag}(z,T-z)$, hence $D_T(z)\ge1/4$.  The square root is increasing and concave on this domain, and $x\mapsto g(x-1/2)$ is increasing and concave for $x\ge1/2$.  Hence $G_T$ is concave on the complete physically admissible $z$ interval, and so is $F$.
Stationarity of Eq.~\eqref{eq:s-gauss-objective} gives
\begin{equation}
                         G_T'(u_*)=e'(u_*)=-\alpha_*,
\end{equation}
therefore $F'(u_*)=0$.  Concavity implies
\begin{equation}
 F(z)\le F(u_*)=G_T(u_*)-e(u_*)=C_{\chi,\mathrm G}^{(1)}.
\end{equation}
The Gaussian displacement ensemble attains the right-hand side, proving
\begin{equation}
 \boxed{C_\chi^{(1)}(\Phi_{\tau,\omega,y},\bar N)
       =C_{\chi,\mathrm G}^{(1)}(\Phi_{\tau,\omega,y},\bar N)}
 \label{eq:s-thermal-result}
\end{equation}
throughout the low-energy one-quadrature branch for every $\tau\ne0,1$ and every physical $y\ge|1-\tau|/2$.  The zero-energy case is immediate.

\section{Threshold, water filling, and the additive-noise limit}

For arbitrary $y$, the threshold at which the resonant seed $s=\omega$ becomes compatible with an isotropic average output is obtained by equating the two output variances:
\begin{equation}
 \boxed{\bar N_{\mathrm{thr}}=
 \frac{1}{2\omega}\left[1+\frac{y}{k}(1-\omega^2)\right]-\frac12.}
 \label{eq:s-threshold}
\end{equation}
Above threshold, the resonant seed has covariance
$\gamma(\omega)=\frac12\operatorname{diag}(\omega^{-1},\omega)$.  Gaussian input and output squeezes transform the channel, for the purpose of unconstrained minimum output entropy, to a phase-insensitive channel with noise $yI$.  The Gaussian optimizer theorem and output majorization results~\cite{Mari2014,Giovannetti2015,DePalma2016} therefore give the unrestricted minimum output entropy
\begin{equation}
                         S_{\min}=g\!\left(\frac{k}{2}+y-\frac12\right).
\end{equation}
The entropy maximum saturates the available photon budget.  At the resulting fixed output trace, the maximum entropy occurs at an isotropic output.  Above Eq.~\eqref{eq:s-threshold} the corresponding modulation covariance is positive semidefinite, so the independent maximum and minimum are simultaneously achievable.  Thus
\begin{equation}
 \boxed{C_\chi^{(1)}=
 g\!\left[k\!\left(\bar N+\frac12\right)+\frac{\Tr Y-1}{2}\right]
 -g\!\left(\frac{k}{2}+y-\frac12\right).}
 \label{eq:s-highcapacity}
\end{equation}
This is the known water-filling expression, now joined to the unrestricted low-energy result.

The additive-noise line $\tau=1$, $y>0$ follows by continuity.  Choose $\tau_n\to1$ with $\tau_n\ne1$ and keep $y,\omega$ fixed.  The matrices $X_{\tau_n}$ and $Y$ converge to those of $\Phi^{\rm F}_{1,\omega,y}$.  On every finite input-energy set the corresponding Gaussian channels converge strongly, with a uniform output-energy bound.  Proposition~6 of Ref.~\cite{ShirokovECD2018} gives a continuity bound for the energy-constrained Holevo capacity under precisely such a uniform output-energy bound, and hence the capacity is continuous along this strong-convergence sequence.  Since Eq.~\eqref{eq:s-thermal-result} holds for every $n$ and the explicit Gaussian variational problem is continuous in $\tau$, the equality passes to the limit.  The identity endpoint $\tau=1,y=0$ is immediate, and the fiducial $\tau=0$ channel has input-independent output and zero Holevo capacity.

Combining the low- and high-energy branches, for every fiducial channel and $k>0$,
\begin{equation}
 C_\chi^{(1)}=
 \max_{s_-\le s\le s_+}
 \left\{h(k\avg V(s,\bar N)+Y)-h(k\gamma(s)+Y)\right\}
 \label{eq:s-lowcapacity}
\end{equation}
for $0\le\bar N<\bar N_{\mathrm{thr}}$, with
$h(V)=g(\sqrt{\det V}-1/2)$ and
$\avg V(s,\bar N)=\operatorname{diag}((2s)^{-1},2\bar N+1-(2s)^{-1})$, while Eq.~\eqref{eq:s-highcapacity} holds above threshold.

\section{From fiducial to arbitrary single-mode Gaussian channels}

Ref.~\cite{Schaefer2013} proves that every full-rank single-mode Gaussian channel with $\tau=\det X\ne0$ and $y=\sqrt{\det Y}>0$ can be written as
\begin{equation}
                         \Phi=\mathcal U_M\circ\Phi^{\rm F}_{\tau,\omega,y}\circ\mathcal U_\Theta,
 \label{eq:s-fiducial-decomp}
\end{equation}
where $\Theta$ is a phase-space rotation, $M$ is a symplectic output transformation, and $\mathcal U_S$ denotes the Gaussian unitary channel associated with a symplectic transformation $S$.  Thus $\mathcal U_\Theta$ is a passive input phase rotation and $\mathcal U_M$ is an output Gaussian unitary.  The input rotation preserves photon number and the output unitary preserves all output entropies.  Therefore both the unrestricted and Gaussian-restricted energy-constrained one-shot Holevo capacities are invariant under Eq.~\eqref{eq:s-fiducial-decomp}.  Equation~\eqref{eq:s-thermal-result}, together with the water-filling branch, proves Gaussian optimality for every full-rank single-mode Gaussian channel.

For completeness, the lower-rank classes can be included without assigning them the same fiducial invariants.  A fixed affine output displacement has already been removed, since it changes neither output entropies nor Holevo information.  Let $\Phi=(X,Y)$ be an arbitrary single-mode Gaussian channel.  The strictly completely-positive, full-rank Gaussian channels are dense in the one-mode Gaussian-channel cone: one may choose full-rank $X_n\to X$ and add a vanishing positive noise perturbation to obtain full-rank $Y_n\to Y$ satisfying the complete-positivity inequality strictly.  Hence there is a sequence $\Phi_n=(X_n,Y_n)$ of full-rank Gaussian channels converging to $\Phi$.

Convergence of the finite matrices $X_n,Y_n$ implies strong convergence of the corresponding Gaussian channels.  Moreover, for every fixed input photon bound $\bar N$, the output second moments are uniformly bounded because $X_n$ and $Y_n$ remain bounded.  Since the photon Hamiltonian satisfies the Gibbs condition and the output energy is uniformly bounded on the constrained input set, Proposition~6 of Ref.~\cite{ShirokovECD2018} applies and the energy-constrained Holevo capacity is continuous along this sequence:
\begin{equation}
 C_\chi^{(1)}(\Phi_n,\bar N)\longrightarrow C_\chi^{(1)}(\Phi,\bar N).
 \label{eq:s-unrestricted-cont}
\end{equation}
To avoid assuming any structural reduction for an arbitrary Gaussian ensemble at a singular channel, introduce the smaller class $C_{\mathrm{disp,G}}^{(1)}$: a centered pure Gaussian seed with covariance $\gamma$, modulated by a centered classical Gaussian displacement distribution to an average covariance $V$.  Explicitly,
\begin{equation}
 C_{\mathrm{disp,G}}^{(1)}(\Phi,\bar N)=
 \max_{\substack{\det\gamma=1/4,\;V-\gamma\ge0\\ \Tr V\le2\bar N+1}}
 \{h(XVX^T+Y)-h(X\gamma X^T+Y)\}.
 \label{eq:s-dispG}
\end{equation}
The feasible covariance set is compact and the objective is continuous, so this quantity is continuous in the finite matrices $X,Y$:
\begin{equation}
 C_{\mathrm{disp,G}}^{(1)}(\Phi_n,\bar N)\longrightarrow
 C_{\mathrm{disp,G}}^{(1)}(\Phi,\bar N).
 \label{eq:s-disp-cont}
\end{equation}
For every full-rank $\Phi_n$, the fiducial proof above is achieved by precisely such a Gaussian displacement ensemble; hence
\begin{equation}
 C_{\mathrm{disp,G}}^{(1)}(\Phi_n,\bar N)=C_\chi^{(1)}(\Phi_n,\bar N).
\end{equation}
Taking $n\to\infty$ in this equality and Eq.~\eqref{eq:s-unrestricted-cont} gives
$C_{\mathrm{disp,G}}^{(1)}(\Phi,\bar N)=C_\chi^{(1)}(\Phi,\bar N)$.  Finally,
\begin{equation}
 C_{\mathrm{disp,G}}^{(1)}\le C_{\chi,\mathrm G}^{(1)}\le C_\chi^{(1)}
\end{equation}
forces all three quantities to coincide.  Consequently
\begin{equation}
 \boxed{C_\chi^{(1)}(\Phi,\bar N)=C_{\chi,\mathrm G}^{(1)}(\Phi,\bar N)}
 \label{eq:s-all-channels}
\end{equation}
for every single-mode Gaussian channel, including the lower-rank canonical classes.  Replacer and identity endpoints may of course be handled directly.

\section{Numerical reproduction and consistency checks}

Figure 1 of the Letter uses a genuinely thermal attenuator
\begin{equation}
                         \tau=0.60,\qquad \omega=0.40,\qquad y=0.30.
\end{equation}
The quantum-limited noise at this transmissivity would be $|1-\tau|/2=0.20$, so the example lies strictly in the mixed-environment regime.  Equation~\eqref{eq:s-threshold} gives
\begin{equation}
                         \bar N_{\mathrm{thr}}=1.275.
\end{equation}
For each $\bar N<\bar N_{\mathrm{thr}}$, bounded scalar optimization of Eq.~\eqref{eq:s-lowcapacity} gives the following values:
\begin{center}
\begin{tabular}{@{}ccc@{}}
\toprule
$\bar N$ & $C_\chi^{(1)}$ (bits/use) & $s_\star$\\
\midrule
0.000 & 0.000000 & 1.000000\\
0.250 & 0.523316 & 0.716871\\
0.500 & 0.872327 & 0.580066\\
1.000 & 1.334509 & 0.443251\\
1.275 & 1.516553 & 0.400000\\
2.000 & 1.895195 & 0.400000\\
3.000 & 2.291769 & 0.400000\\
\bottomrule
\end{tabular}
\end{center}
At $\bar N=0.5$, the optimizer has
\begin{equation}
 s_\star=0.580066,\quad
 \mu_\star=0.610391,
 \quad q_\star=0.315303,
 \quad t_p=0.557502,
 \quad \alpha_\star=0.216648.
\end{equation}
Numerical differentiation independently gives
$-e'(u_\star)=0.216648$ and
$G_T'(u_\star)=e'(u_\star)$ to the displayed precision.  Scanning the supporting function $F(z)$ in Eq.~\eqref{eq:s-F} over the complete feasible covariance interval places its unique maximum at $z=u_\star$.

The verification and plotting scripts supplied with the manuscript were prepared with assistance from ChatGPT (GPT-5.6 Sol, OpenAI).  The tested analytic identities are stated explicitly above; the code paths and numerical outputs were reviewed against those identities and the tabulated values before inclusion.

The analytic structure also provides the following checks:
\begin{enumerate}
 \item At $\bar N=0$, the feasible interval collapses and $C_\chi^{(1)}=0$.
 \item At threshold, $s_\star\to\omega$, $\alpha_\star\to0$, and the average output becomes isotropic.
 \item For $\omega=1$, $\bar N_{\mathrm{thr}}=0$ and only the phase-insensitive water-filling branch remains.
 \item At the quantum-limited value $y=|1-\tau|/2$, the result reduces to the pure-environment formulas of the original pointwise theorem.
 \item Random numerical tests over attenuating, amplifying, and phase-conjugating parameters verify the identity $\lambda_{q,t}\beta_{\rm W}=-e'(u)$ and the supporting inequality $F(z)\le F(u_\star)$; the supplied script \texttt{proof\_checks.py} reproduces these checks.
\end{enumerate}

\section{Why the result remains one-shot}

For $n$ channel uses and average input $\rho_{A^n}$, let $E$ denote the full complementary output associated with one use (one mode for a pure physical environment and two modes after purification of a mixed environment).  A Stinespring dilation gives
\begin{equation}
 \chi_{\Phi^{\otimes n}}(\rho_{A^n})
 =S(\Omega_{B^n})-\EF(B^n{:}E^n)_\Omega.
 \label{eq:s-nuse}
\end{equation}
The pointwise pure-environment argument would require fixed-covariance formation extremality for arbitrary $n\times n$ bipartitions, which is not supplied by the two-mode theorem or its bisymmetric extension~\cite{Adesso2026,Serafini2005}.  The thermal proof above is different but still single-use: its supporting EPR witness is tangent to a one-mode Gaussian optimizer and annihilates one modulation direction.  With correlated $n$-use inputs there are cross-use covariance directions and entanglement structures that are not controlled by the direct sum of these single-use witnesses.

Consequently the regularized classical capacity
\begin{equation}
 C(\Phi,\bar N)=\lim_{n\to\infty}\frac1n
 C_\chi^{(1)}(\Phi^{\otimes n},n\bar N)
\end{equation}
is not determined here.  Proving additivity would require an additional multimode statement, such as a suitable global formation witness or a strong-superadditivity principle.
\endgroup
\end{document}